\documentclass[10pt]{article}

\usepackage{amsthm,amssymb,stmaryrd}
\usepackage{dsfont}
\usepackage{epsfig}
\usepackage{mathrsfs}
\usepackage{amsmath}
\usepackage{cases}

\newcommand{\vv} { {\bf v}}

\newcommand{\vx} { {\bf x}}

\def\si{\sigma}

\def\ra{\rangle}

\newcommand{\bN} { {\mathbb{N}}}
\newcommand{\bC} { {\mathbb{C}}}
\newcommand{\bQ} { {\mathbb{Q}}}
\newcommand{\bZ} { {\mathbb{Z}}}

\newcommand{\bF} { {\mathbb{F}}}
\newcommand{\bK} { {\mathbb{K}}}
\newcommand{\bE} { {\mathbb{E}}}

\def\pa{{\partial}}

\def\CX{{\mathbb C}}
\def\QX{{\mathbb Q}}

\def\bN{{\mathbb N}}
\def\bZ{{\mathbb Z}}

\def\d{{\partial}}
\def\fraka{{\mathfrak a}}

\def\frakm{{\mathfrak m}}

\def\frakp{{\mathfrak p}}
\def\frakq{{\mathfrak q}}

\def\d{{\partial}}

\def\Gal{{\rm Gal}}

\def\min{{\rm min}}

\def\pa{{\partial}}

\def\bfx{\mathbf x}
\def\bfc{\mathbf c}

\newcommand{\Mat} {{\rm Mat}}

\newtheorem{thm}{Theorem}[section]

\newtheorem{lem}[thm]{Lemma}
\newtheorem{prop}[thm]{Proposition}
\newtheorem{definition}[thm]{Definition}

\newtheorem{remark}[thm]{Remark}
\newtheorem{exam}[thm]{Example}

\newtheorem{algorithm}[thm]{Algorithm}
\newtheorem{problem}[thm]{Problem}

\makeatletter \@addtoreset{equation}{section}
\begin{document}

\title{Separability Problems in Creative Telescoping
%Breaking Through the Bivariate Barrier
\thanks{S.\ Chen was partially supported by the NSFC
grants 11871067, 11688101, the Fund of the Youth Innovation Promotion Association, CAS,
and the National Key Research and Development Project 2020YFA0712300. R. Feng was partially supported by the NSFC
grants  11771433, 11688101, and Beijing Natural Science Foundation under Grant Z190004. P.\ Ma was partial supported by the NSFC
grants 11871067.  M.F. Singer was partially supported by by a grant from the Simons Foundation (No. 349357, Michael Singer)}}

\author{
Shaoshi Chen$^{a, b}$, Ruyong Feng$^{a, b}$, \\
\bigskip
Pingchuan Ma$^{a, b}$, and Michael F.\ Singer$^{c}$\\
$^a$KLMM,\, Academy of Mathematics and Systems Science, \\ Chinese Academy of Sciences, \\Beijing, 100190, China\\
$^b$School of Mathematical Sciences, \\University of Chinese Academy of Sciences,\\ Beijing 100049, (China)\\
$^c$Department of Mathematics, \\
North Carolina State University, \\
\medskip
Raleigh, NC 27695, (USA)\\
{\sf schen@amss.ac.cn,  ryfeng@amss.ac.cn}\\
{\sf mapingchuan15@mails.ucas.ac.cn, singer@math.ncsu.edu}
%\\ {\sf wangwang@mail.nankai.edu.cn}
%Preliminary notes
}

%\author{\medskip
%Shaoshi Chen, \, Lixin Du, \, Chaochao Zhu \\
%\smallskip
%       \affaddr{$^1$KLMM, Academy of Mathematics and Systems Science, Chinese Academy of Sciences, Beijing 100190, (China)}\\
%       \smallskip
%       \affaddr{$^2$School of Mathematical Sciences, University of Chinese Academy of Sciences, Beijing 100049, (China)}\\
%       \smallskip
%      \email{schen@amss.ac.cn, dulixin17@mails.ucas.ac.cn, chaochaozhu@139.com}
%%Preliminary notes

\maketitle

\begin{abstract}
For given multivariate functions specified by algebraic, differential or difference equations,
the separability problem is to decide whether they satisfy linear differential or difference equations
in one variable.  In this paper, we will explain how separability problems arise naturally in creative telescoping
and present some criteria for testing the separability for several classes of special functions,
including rational functions, hyperexponential functions, hypergeometric terms, and algebraic functions.
\end{abstract}

%\keywords{Creative telescoping, Separable functions, Separation of variables, Zeilberger's algorithm}

\section{Introduction}\label{SECT:intro}
The method of separation of variables has been used widely in solving differential equations~\cite{Miller1977}.
In order to solve the one-dimensional heat equation
\[\frac{\partial y}{\partial t} - c\frac{\partial^2 y}{\partial x^2} = 0, \, \, \text{where $c \in \bC$},\]
together with the boundary conditions $y(t, 0) = y(t, L) = 0$. One can try to find a nonzero solution of the form
\[y = u(t)v(x),\]
and then substitute this form into the equation to get
\[\frac{\frac{\partial u(t)}{\pa t}}{u} = c\frac{\frac{\partial^2 v(x)}{\pa x^2}}{v}.\]
Since both sides only depend on one variable, there exits some constant $\lambda \in \bC$ such that
\[\frac{\partial u}{\partial t} - \lambda u = 0 \quad \text{and} \quad c \frac{\partial^2 v}{\partial x^2} - \lambda v = 0. \]
Note that the above two equations are also satisfied by $y = u(t) v(x)$, which are linear differential equation in only one variable. After solving these special equations with the boundary conditions into account, a special solution of the heat equation can be given as
\begin{equation}\label{EQ:heat}
y(t, x) = \sum_{n=1}^\infty d_n \sin\left(\frac{n\pi x}{L}\right)\exp \left(-\frac{n^2\pi^2 c t}{L^2}\right),
\end{equation}
where $d_n\in \bC$ are coefficients determined by the initial conditions.
Motivated by this example, one would ask the following natural question.

\begin{problem}[Separability Problem]\label{PROB:sep}
Given a multivariate function specified by certain equations (e.g.\ algebraic, differential or difference equations),
decide whether this function satisfies linear differential or difference equations
in one variable.
\end{problem}
To make the problem more tractable, we will consider some special classes of functions, such as rational functions, algebraic functions, hyperexponential functions and hypergeometric terms etc..
The main goal of this paper is to show the close connection between the separability problem
and Zeilberger's method of creative telescoping~\cite{Wilf1992, Zeilberger1991}.

The remainder of this paper is organized as follows. We specify the separability problem and the existence problem of telescopers
precisely in Section~\ref{SECT:preli} together with the definition of orders and (local) dispersions of rational functions. After this,
we explain how the separability problems arise naturally in creative telescoping for rational functions in Section~\ref{SECT:rational},
hyperexponential functions and hypergeometric terms in Section~\ref{SECT:hyper}, and for algebraic functions in Section~\ref{SECT:alg}.
Separability criteria will be given for these classes of special functions. We then conclude our paper with some comments on
the separability  problem on D-finite functions and P-recursive sequences.

%and recall different types of reductions that are used in testing
%the exactness of bivariate rational functions in Section~\ref{SECT:red}. Existence criteria are given
%for four types of telescopers for rational functions in three variables in Section~\ref{SECT:criteria}.

%\section{Creative Telescoping} \label{SECT:ct}
%The remainder of this paper is organized as follows. We define the existence problem of telescopers
%precisely in Section~\ref{SECT:preli} and recall different types of reductions that are used in testing
%the exactness of bivariate rational functions in Section~\ref{SECT:red}. Existence criteria are given
%for four types of telescopers for rational functions in three variables in Section~\ref{SECT:criteria}.

%\bigskip
\section{Preliminaries}\label{SECT:preli}
Let~$\bF$ be a field of characteristic zero and let ${\bE = \bF(t, \vx)}$
be the field of rational functions in~$t$ and $\vx= (x_1, \ldots, x_m)$ over~$\bF$.
Let $\delta_t, \delta_{x_i}$ be the usual partial derivations
$\pa/\pa_t, \pa/\pa_{x_i}$ with $x_i \in \{x_1, \ldots, x_m\}$, respectively. The shift operators~$\si_t$ and $\si_{x_i}$ on~$\bE$ are defined
as the $\bF$-automorphisms such that for any $f\in \bE$, $\si_t(f(t, \vx)) = f(t+1, \vx)$ and
\[\si_{x_i}(f(t, \vx)) = f(t, x_1, \ldots, x_{i-1}, x_i+1, x_{i+1}, \ldots, x_m).\]
%Let $q\in \bK\setminus \{0\}$ be fixed throughout this paper, which is not a root of unity. The $q$-shift operators~$\tau_x, \tau_y, \tau_z$ on~$\bE$ are defined
%by $\tau_x(f)=f(qx, y, z)$, $\tau_y(f)=f(x, qy, z)$, and $\tau_z(f)=f(x, y, qz)$ for $f\in \bE$, respectively.
The ring of linear functional operators in $t$ and $\vx$ over $\bE$ is denoted
by $\bE\langle \partial_{t}, \partial_{\vx} \rangle $, where $\partial_{\vx} = (\pa_{x_1}, \ldots, \pa_{x_m})$ and $\partial_v$ with $v\in \{t, \vx\}$ is either the derivation $D_v$ such that $D_v f = f D_v + \delta_v(f)$
or the shift operator $S_v$ such that $S_v f = \si_v(f) S_v$ for any $f\in \bE$, and $\partial_t$ and $\partial_{x_i}$ commute.
For $v\in \{t, \vx\}$, we let $\Delta_v$ denote the difference operator $S_v-{\bf 1}$, where ${\bf 1}$ stands for the identity map on $\bE$.
Abusing notation, we let $\delta_v$ and $\si_v$ denote arbitrary extensions of $\delta_v$ and $\si_v$ to derivation and $\overline\bF$-automorphism
of $\overline{\bE}$, the algebraic closure of $\bE$. The functions we consider will be in certain differential or difference extension of $\bE$, which is also
an $\bE\langle \partial_{t}, \partial_{\vx} \rangle $-module via the action defined by simply interpreting  $D_v, S_v$ by  $\delta_v, \si_v$, respectively,  for $v\in \{t, \vx\}$.
The ring $\bF(t)\langle \partial_t\rangle$ is a subring of $\bE\langle \partial_{t}, \partial_{\vx}\ra$ that is also
a left Euclidean domain. Efficient algorithms for
basic operations in $\bF(t)\langle \partial_t\rangle$, such as computing
the least common left multiple (LCLM) of operators, have
been developed in~\cite{BronsteinPetkovsek1996, AbramovLeLi2005}.
%\begin{lem}\label{LEM:lclm}
%For an operator $L = \sum_{i=0}^\rho e_i D_x^i\in \overline{\bK(x)}\la D_x \ra$ with $e_{\rho} =1$,
%we let $\bF$ be a finite normal extension of $\bK(x)$ containing the coefficients $e_i$'s
%and $G$ be the Galois group of $\bF$ over~${\bK(x)}$. Let $T$ be the LCLM of the operators $\si(L) = \sum_{i=0}^\rho \si(e_i)D_x^i$
%for all $\si\in G$. Then $T$ belongs to $\bK(x)\la D_x\ra$.
%\end{lem}
%\begin{proof}
%It suffices to show that $\tau(T) = T$ for all $\tau\in  G$.
%Since $D_x$ commutes with any isomorphism in $G$~by~\cite[Theorem 3.2.4~(i)]{BronsteinBook}, we have $\tau(L_1 L_2) = \tau(L_1)\tau(L_2)$
%for all $L_1, L_2\in \bF\la D_x \ra$.
%For each $\si \in G$, we have $T = P_{\si} \si(L)$ for some $P_\si \in \bF\la D_x \ra$, which implies that
%$\tau(\si(L))$ divides $\tau(T)$. When $\si$ runs through all elements of $G$, so does $\tau\si$.
%Hence $\tau(T)$ is also a common left multiple of the operators $\si(L)$
%for all $\si\in G$. Since $\tau(T)$ and $T$ are both monic and of the same degree in $D_x$,
%we get $\tau(T)=T$.
%\end{proof}
%\begin{rem}
%The above assertion is not true in the shift case. For example, take $L = S_x + \sqrt{x}$.
%The LCLM of $L$ and its conjugation $S_x -\sqrt{x}$ is $S_x^2 - \sqrt{x(x+1)}$, which is not in $\bK(x)\la S_x\ra$.
%\end{rem}

\begin{definition}[Separable functions]\label{DEF:sep}
Let $\mathfrak{M}$ be an $\bE\langle \partial_{t}, \pa_{\vx} \rangle $-module and $f\in \mathfrak{M}$. We say that $f(t, \vx)$
is \emph{$\partial_t$-separable} if there exists a nonzero  $L\in \bF(t)\langle \partial_t\rangle$ such that $L(f)=0$.
\end{definition}
As an example, the special solution~\eqref{EQ:heat} of the one-dimensional heat equation is both $D_t$-separable and $D_x$-separable.
Note that $\partial_t$-separable functions are just the D-finite functions in the differential case and
the P-recursive sequences in the shift case, which are both introduced in~\cite{Stanley1980}. By the closure properties of D-finite functions and
P-recursive sequences, we have the same closure properties for  $\partial_t$-separable functions.
\begin{prop}\label{PROP:closure}
Let $\mathfrak{M}$ be an $\bE\langle \partial_{t}, \pa_{\vx} \rangle $-module. If  $f, g\in \mathfrak{M}$ are $\partial_t$-separable, so are
$f + g, f\cdot g$, and $a\cdot f$ for all $a\in \bF(t)$.
\end{prop}
%\begin{proof}
%For the proof, see~\cite{Stanley1980}.
%\end{proof}
We will focus on the separability problem on function in an $\bE\langle \partial_{t}, \pa_{\vx} \rangle $-module.
%\begin{problem}[Separability Problem]\label{PROB:sep}
%Given a  function~$f(t, \vx)$ in an $\bE\langle \partial_{t}, \pa_{\vx} \rangle $-module, decide whether $f$ is $\partial_t$-separable or not.
%\end{problem}

\begin{definition}[Creative telescoping]\label{DEF:telescoper}
Let $\mathfrak{M}$ be an $\bE\langle \partial_{t}, \pa_{\vx} \rangle $-module and $f\in \mathfrak{M}$. A nonzero operator $L\in \bF(t)\langle \partial_t\rangle$
is called a \emph{telescoper} of type $(\partial_t, \partial_{\vx})$ for $f$ if there exist $Q_1, \ldots, Q_m\in \bE\langle \partial_{t}, \pa_{\vx} \rangle $
such that
\begin{equation}\label{EQ:telescoper}
L(t, \pa_t)(f) = \partial_{x_1}(Q_1(f)) + \cdots + \partial_{x_m}(Q_m(f)),
\end{equation}
where $\partial_t \in \{D_t, S_t\}$ and $\partial_{x_i} \in \{D_{x_i}, \Delta_{x_i}\}$.
%The operators $Q_1, \ldots, Q_m$
%are called \emph{certificates} of $L$ in $\mathfrak{M}$.
\end{definition}

The central problem in the Wilf-Zeilberger theory of automatic proving of special-function identities is
related to the existence and the computation of telescopers for special functions. In the next sections, we will show that this central problem on creative telescoping is
closely connected to the separability problem
on the corresponding class of special functions.

%\begin{problem}\label{PROB:telescoper}
%Given~$f\in \bK(x, y, z)$, decide whether $f$ has a telescoper of type $(\partial_x, \partial_{y}, \partial_{z})$.
%\end{problem}

Let $V = (V_1, \ldots, V_s)$ be any set partition of the variables $\vv= \{t, x_1, \ldots, x_m\}$.
A rational function $f\in \bF(t, \vx)$ is said to be \emph{split} with respect to the partition $V$
if $f = f_1\cdots f_s$ with $f_i\in \bF(V_i)$ and be \emph{semi-split} with respect to $V$ if there are split functions $g_j\in  \bF(t, \vx)$ such that
$f =  g_1 + \cdots + g_n$. By definition, we have $f = p/q$ with $p, q \in \bF[t, \vx]$ and $\gcd(p, q)=1$ is semi-split with respect to the partition $V$
if and only if the denominator $q$ is a split polynomial with respect to the partition $V$.
Split rational functions will be used to describe the separability of given functions.

Let $\bK = \bF(\vx)$ and $p\in \bK[t]$ be an irreducible polynomial in $t$. For any $f\in \bK(t)$, we can write $f = p^m a/b$, where $m \in \bZ, a, b \in \bK[t]$ with $\gcd(a, b) = 1$ and $p\nmid ab$.
Conventionally, we set $\nu_p(0)= +\infty$.
The integer $m$ is called the \emph{order} of $f$ at $p$, denoted by $\nu_p(f)$.
We collect some basic properties of valuations as follows and refer to~\cite[Chapter 4]{BronsteinBook} for their proofs.

\begin{prop}\label{PROP:val}
 Let $f, g \in \bK(t)$ and $p \in \bK[t]$ be an irreducible polynomial. Then,
 \begin{itemize}
 \item[$(i)$] $\nu_p(fg) = \nu_p(f) + \nu_p(g)$.
 \item[$(ii)$] $\nu_p(f+g) \geq \min\{\nu_p(f), \nu_p(g)\}$ and equality holds if $\nu_p(f) \neq \nu_p(g)$.
 \item[$(iii)$] If $\nu_p(f) \neq 0$, then $\nu_p(D_t(f)) = \nu_p(f)-1$. In particular, for any $i\in \bN$, $\nu_p(D_t^i(f)) = \nu_p(f)-i$ if $\nu_p(f)<0$.
 \end{itemize}
\end{prop}

The dispersion introduced by Abramov in~\cite{Abramov1971} can be viewed as a shift analogue of the order.
For any polynomial $u\in \bK[t]$ with $\deg_t(u)\geq 1$, the \emph{dispersion} of $u$, denoted by $\text{dis}(u)$, is defined as
$\max\{k\in \bN \mid \gcd(u, \si_t^k(u)) \neq 1\}$, which is  the maximal integer root-distance $|\alpha - \beta|$ with $\alpha, \beta$ being roots of $u$ in $\bar K$.
Define $\text{dis}(u)=0$ if $u \in K\setminus \{0\}$ and $\text{dis}(0)=+\infty$.  For a rational function $f = a/b\in \bK(t)$ with $a, b\in \bK[t]$ and $\gcd(a, b)=1$,
define $\text{dis}(f) = \text{dis}(b)$. For later use, we introduce a local version of Abramov's dispersion.
%the set $[p]_{\si_t} := \{\si_t^i(p)\mid i\in \bZ\}$ is called the $\si_t$-orbit at $p$.
Let $p\in \bK[t]$ be an irreducible polynomial. If $\si_t^i(p) \mid u$ for some $i\in \bZ$, the \emph{local dispersion} of $u$ at $p$, denoted by $\text{dis}_p(u)$,  is defined as
the maximal integer distance $|i-j|$ with $i, j\in \bZ$ satisfying $\si_t^i(p)\mid u$ and $\si_t^j(p)\mid u$; otherwise we define $\text{dis}_p(u)=0$.
Conventionally, we set $\text{dis}_p(0) = +\infty$. For a rational function $f = a/b\in \bK(t)$ with $a, b\in \bK[t]$ and $\gcd(a, b)=1$, we also define $\text{dis}_p(f) = \text{dis}_p(b)$.
By definition,  we have
\[\text{dis}(u) = \max \{\text{dis}_p(u)\mid \text{$p$ is an irreducible factor of $u$}\}.\]
The set $\{\si_t^i(p)\mid i\in \bZ\}$ is called the $\si_t$-orbit at $p$, denoted by $[p]_{\si_t}$.
Note that $\text{dis}_p(u) = \text{dis}_q(u)$ if $q \in [p]_{\si_t}$. So we can define the local dispersion of a rational function $f$
at a  $\si_t$-orbit at $p$, denoted by $\text{dis}_{[p]_{\si_t}}(f)$.
%We remark that the original definition of dispersions in~\cite{Abramov1971} for a polynomial can be viewed as a global version of the dispersion defined as above,
%which are the maximal integer distance between roots of the given polynomial.
%Some basic properties of local dispersions are collected below and
\begin{exam}
Let $u = x(x+1)(x-5)(x^2+1)(x^2+4x+5)\in \bQ[x]$. Then we have  $\text{dis}_x(u) = 6$ and $\text{dis}_{x^2+1}(u) = 2$.
Abramov's dispersion of $u$ is then equal to $6$.
\end{exam}

We now  shows how the local dispersions change under the action of linear recurrence operators,
which was first proved for Abramov's dispersions in~\cite{Abramov1971, Abramov1974} and~\cite[Section 3.1]{Paule1995}.

\begin{lem}\label{LEM:dis}
Let $f=a/b \in \bK(t)$ with $a, b\in \bK[t]$ and $\gcd(a, b)=1$ and let $p \in \bK[t]$ be an irreducible factor of $b$.
Let $L = \sum_{i=0}^\rho \ell_iS_t^i  \in \bK[t]\langle S_t\rangle$ be such that $\ell_{\rho}\ell_0 \neq 0$ and $\sigma^i_t(p)$ does not divide $\ell_\rho \ell_0$ for any $i\in \bZ$.
Then $\text{dis}_p(L(f)) = \text{dis}_p(f) + \rho$. In particular, $\text{dis}_p(\Delta_t(f)) = \text{dis}_p(f) + 1$.
\end{lem}
\begin{proof}
Let $d = \text{dis}_p(b)$.  Without loss of generality, we may assume that $p\mid b$ but $\si_t^i(p)\nmid b$ for any $i<0$.
Since $\gcd(a, b)=1$ and $\si_t$ is a $\bK$-automorphism of $\bK[t]$, we have $\gcd(\si_t^i(a), \si_t^i(b))=1$ for any $i \in \bZ$.
Applying $L$ to $f$ yields
\[L(f) =\sum_{i=0}^\rho  \ell_i \si_t^i\left(\frac{a}{b}\right) = \frac{\sum_{i=0}^\rho \ell_i \si_t^i(a)u_i}{u},\]
where $u=b\si_t(b)\cdots \si_t^\rho(b)$ and $u_i = u/\si_t^i(b)$. Write $L(f) = A/B$ with $A, B\in \bK[t]$ and $\gcd(A, B)=1$.
Then $B \mid u$ and $\text{dis}_p(L(f)) = \text{dis}_p(B)$ by definition.
Since $\si_t^i(p)\nmid \ell_0$ and $\si_t^i(p)\nmid \ell_{\rho}$ for any $i\in \bZ$, we have
both $p$ and $\si_t^{d + \rho}(p)$ do not divide the sum $\sum_{i=0}^\rho \ell_i \si_t^i(a)u_i$, but they divide $u$. So $p\mid B$ and
$\si_t^{d + \rho}(p)\mid B$, which implies that $\text{dis}_p(B)\geq d+\rho$. Since $B\mid u$, we have
 $\text{dis}_p(B)\leq \text{dis}_p(u)= d+\rho$. Therefore, $\text{dis}_p(L(f)) = d+\rho$.
\end{proof}

\section{The rational case}\label{SECT:rational}
We first explain how the existence problem of telescopers for rational functions is naturally
connected to the separability problem on this class of functions.
Let $f(t, x)$ be a bivariate rational function in $\bF(t, x)$. By the Ostrogradsky-Hermite reduction~\cite{Ostrogradsky1845, Hermite1872}, we can decompose $f$
into the form
\[f = D_x(g) + \frac{a}{b},\]
where $g\in \bF(t, x)$ and $a, b\in \bF(t)[x]$ with $\gcd(a, b)=1$, $\deg_x(a)<\deg_x(b)$ and $b$ being squarefree in $x$ over $\bF(t)$. Moreover,  $f = D_x(h)$ for some $h\in \bF(t, x)$ if and only if $a=0$.
Then $f$ has a telescoper of type $(S_t, D_x)$ if and only if $a/b$ does. Applying a nonzero operator $L = \sum_{i=0}^\rho \ell_i S_t^i\in \bF(t)\langle S_t\rangle$ to $a/b$ yields
\[L\left(\frac{a}{b}\right) = \sum_{i=0}^\rho \ell_i(t) \si_t^i \left(\frac{a}{b}\right) = \sum_{i=0}^\rho \frac{\ell_i(t)a(t+i, x)}{b(t+i, x)} = \frac{p}{q}, \]
 where $p, q \in \bF[t, x]$ with $\gcd(p, q)=1$. Since the shift operator $S_t$ is an $\bF(x)$-automorphism and preserves the degrees in $t$ and $x$,
 we have $b(t+i, x)$ is squarefree in $x$ over $\bF(t)$ for any $i\in \bN$ and $\deg_x(a(t+i, x)) < \deg_x(b(t+i, x))$. So $\deg_x(p)< \deg_x(q)$ and
 $q$ is also squarefree in $x$ over $\bF(t)$.   This implies the operator $L$ is a telescoper of type $(S_t, D_x)$ for $a/b$, i.e., $L(a/b) = D_x(g)$ for some $g\in \bF(t, x)$
 if and only if $p =0$, i.e., $L(a/b) = 0$. Therefore, we conclude that $f$ has a telescoper of type $(S_t, D_x)$ if and only $a/b$ is $S_t$-separable.

We can also consider telescopers of type $(D_t, S_x)$. By Abramov's reduction~\cite{Abramov1975, Abramov1995b}, we can decompose $f\in \bF(t, x)$ into the form
\[f = \Delta_x(g) + \frac{a}{b},\]
where $g\in \bF(t, x)$ and $a, b\in \bF(t)[x]$ with $\gcd(a, b)=1$, $\deg_x(a)<\deg_x(b)$ and $b$ being shift-free in $x$ over $\bF(t)$, i.e., $\gcd(b, \si_x^i(b))=1$ for all nonzero $i\in \bZ$.
Applying a nonzero operator $L = \sum_{i=0}^\rho \ell_i(t) D_t^i\in \bF(t)\langle D_t\rangle$ to $a/b$ yields
\[L\left(\frac{a}{b}\right) = \sum_{i=0}^\rho \ell_i \delta_t^i \left(\frac{a}{b}\right)  = \sum_{i=0}^\rho \frac{\ell_i(t)a_i}{b^{i+1}} = \frac{p}{q}, \]
where $a_i, p, q \in \bF[t, x]$ with $\deg_x(a_i) < (i+1)\deg_x(b)$ and $\gcd(p, q)=1$. Since $b$ is shift-free in $x$, so is $b^i$ for any $i\in \bN$.
Note that any factor of a shift-free polynomial is still shift-free. So $q$ is shift-free and $\deg_x(p)< \deg_x(q)$.
This implies the operator $L$ is a telescoper of type $(D_t, S_x)$ for $a/b$, i.e., $L(a/b) = \Delta_x(g)$ for some $g\in \bF(t, x)$
if and only if $p =0$, i.e., $L(a/b) = 0$. Then we also have  that $f$ has a telescoper of type $(D_t, S_x)$ if and only $a/b$ is $D_t$-separable.

The next theorem characterizes all possible separable rational functions in terms of semi-split rational functions.
\begin{thm}\label{THM:seprat}
A rational function $f \in \bF(t, \vx)$ is $\partial_t$-separable if and only if $f$
is semi-split in $t$ and $\vx$.
\end{thm}
\begin{proof}
Assume that $f$ is semi-split in $t$ and $\vx$. Then $f= a_1 b_1 + \cdots + a_n b_n$,
where $a_i \in \bF(t)$  and $b_i\in \bF(\vx)$ for all $i$ with $1\leq i \leq n$. Since each $a_ib_i$ is annihilated by the operator $L_i := \partial_t - \partial_t(a_i)/a_i \in \bF(t)\langle \pa_t\rangle$,
the rational function $f$ is annihilated by $\text{LCLM}(L_1, \ldots, L_n)$. So $f$ is $\partial_t$-separable.

For the necessity we assume that $f = a/b$ with $a, b \in \bF[t, \vx]$ and $\gcd(a, b)=1$
is $\partial_t$-separable, i.e., there exists a nonzero operator $L = \sum_{i=0}^\rho \ell_i \pa_t^i\in \bF(t)\langle \pa_t\rangle$ with $\ell_\rho \neq 0$
such that $L(f)=0$. It suffices to show that the denominator $b$ is split with respect to $t$ and $\vx$.
Suppose for the sake of contradiction that $b$ is not split.  Then $b$ has at least one irreducible factor $p$ such that
$p$ is not split. Now we proceed by a case distinction according to the type of $\partial_t$. In the case when $\pa_t = D_t$,
we have $\nu_p(\ell_iD_t^i(f))= \nu_p(f)-i$ for each $i$ with $\ell_i \neq 0$, since  $\nu_p(f) < 0$ and $\nu_p(\ell_i) =0$, which implies further that $\nu_p(L(f)) = \nu_p(f) - \rho$ by Proposition~\ref{PROP:val}.
But $\nu_p(L(f))= \nu_p(0) = +\infty$, which leads to an contradiction. In the case when $\pa_t = S_t$, we may always assume that $\ell_i \in \bF[t]$ and
$\ell_0 \neq 0$ since $\si_t$ is an $\bF(\vx)$-automorphism of $\bF(t, \vx)$.
Since $\ell_0$ and $\ell_\rho$ are free of $x$, we have $\si_t^i(p) \nmid \ell_0\ell_\rho$ for any $i\in \bZ$.
By Lemma~\ref{LEM:dis}, we get $\text{dis}_p(L(f)) = \text{dis}_p(f)+\rho < \infty$, which contradicts with $\text{dis}_p(L(f)) = \text{dis}_p(0) = +\infty$.
\end{proof}

\begin{remark}
With the above theorem, we can detect easily the $\partial_t$-separability of rational functions by the computation of contents and derivatives of multivariate polynomials in $t$.
 %The above detected by  algorithms for computing the Gosper form and its differential analogue in~\cite{Gosper1978, Almkvist1990}.
\end{remark}

\section{The Hyperexponential and Hypergeometric \\Cases}\label{SECT:hyper}
The separability problem on hyperexponential functions and hypergeometric terms was first studied in~\cite{LeLi04}, which was later connected
to the existence of parallel telescopers for hyperexponential functions~\cite{Chen2014}.  We motivate this problem by revisiting Zeilberger's algorithm which
computes telescopers for hypergeometric terms (see~\cite[Chapter 6]{PWZbook1996}).

Let $H(t, x)$ be a nonzero hypergeometric term over the rational-function field $\bF(t, x)$, i.e., both $\si_t(H)/H$ and $\si_x(H)/H$ are in $\bF(t, x)$.
If telescopers of type $(S_t, S_x)$ exist for $H$, Zeilberger's algorithm starts from an ansatz: for fixed $\rho \in \bN$, set $L= \sum_{i=0}^\rho \ell_i S_t^i \in \bF(t)\langle S_t\rangle$ with the $\ell_i$'s
being undetermined coefficients. Applying $L$ to $H$ yields
\[T := L(H) = \sum_{i=0}^\rho \ell_i \si_t^i(H) = \sum_{i=0}^\rho \ell_i a_i H = \frac{\sum_{i=0}^\rho \ell_i P_i}{Q}H,\]
where $a_i = \si_t^i(H)/H = P_i/Q \in \bF(t, x)$ with $P_i, Q\in \bF[t, x]$. The second step of Zeilberger's algorithm is computing
the Gosper form of $L(H)$ that gives
\[\frac{\si_x(L(H))}{L(H)} = \frac{\si_x\left(\sum_{i=0}^\rho \ell_i P_i\right)}{\sum_{i=0}^\rho \ell_i P_i} \frac{\si_x(p)}{p} \frac{q}{r},\]
where $(p, q, r)\in \bF(t)[x]^3$ is a Gosper form of the rational function
\[ \frac{Q\si_x(H)}{(\si_x(Q)H)}\]
satisfying that $\gcd(q, \si_x^i(r))=1$ for all $i\in \bZ$.
The last step is finding $\ell_0, \ldots, \ell_\rho \in \bF(t)$, not all zero,  such that the equation
\[\left(\sum_{i=0}^\rho \ell_i P_i\right)p = q \si_x(z) - \si_x^{-1}(r) z.\]
 has a polynomial solution in $\bF(t)[x]$. If so,  then $L = \sum_{i=0}^\rho \ell_iS_t^i $ is a telescoper for $H$.
It may happen that the final choice of the $\ell_i$'s satisfies that $\sum_{i=0}^\rho \ell_i P_i=0$. This  means division by zero
may happen in the second step. To avoid this, we should first detect whether $L(H) = 0$ for some $L\in \bF(t)\langle S_t\rangle$, i.e.,
the separability problem on hypergeometric terms.

The following theorem characterizes all possible separable hyperexponential functions and hypergeometric terms, whose
proof was given in~\cite[Lemma 4]{LeLi04} or in~\cite[Proposition 10]{Chen2014}.
\begin{thm} Let $\mathfrak{M}$ be an $\bE\langle \partial_{t}, \pa_{\vx} \rangle $-module and let $H\in \mathfrak{M}$ be such that
\[\partial_t(H) = a H\,\, \text{and} \,\, \pa_{x_i}(H) = b_i H \,\, \text{with $a, b_i \in \bF(t, \vx)$}. \]
Then we have,
\begin{itemize}
\item[$(i)$] Hyperexponential case: $H$ is $D_t$-separable if and only if there exist $p\in \bF(\vx)[t]$ and $r\in \bF(t)$ such that
\[ a = \frac{\delta_t(p)}{p} + r.\]
\item[$(ii)$] Hypergeometric case: $H$ is $S_t$-separable if and only if there exist $p\in \bF(\vx)[t]$ and $r\in \bF(t)$ such that
\[ a = \frac{\si_t(p)}{p} \cdot r.\]
\end{itemize}
\end{thm}
\begin{remark}
The above form for $\partial_t(H)/H$ can be detected by  algorithms for computing the Gosper form and its differential analogue in~\cite{Gosper1978, Almkvist1990}.
\end{remark}
\section{The Algebraic case}\label{SECT:alg}
In this section, we solve the separability problem on algebraic functions.
We first explain the connection between this problem and the following existence problem of telescopers for rational functions in three variables.
\begin{problem}\label{PROB:DSD}
Given $f\in \bF(t, x, y)$, decide whether there exists a nonzero operator $L\in \bF(t)\langle D_t \rangle$
such that $L(f) = \Delta_x(g) + D_y(h)$ for some $g, h\in \bF(t, x, y)$.
\end{problem}
By applying the Ostrogradsky-Hermite reduction in $y$ and Abramov's reduction in $x$ to  $f\in \bF(t, x, y)$, we get
\[f = \Delta_x(u) + D_y(v) + r \, \, \text{with $r = \sum_{i=1}^I \frac{\alpha_i}{y-\beta_i}$}\]
where $u, v, r \in \bF(t, x, y)$, $\alpha_i, \beta_i \in \overline{\bF(t, x)}$ and $\beta_i$'s are in distinct $\si_x$-orbits. Then $f$
has a telescoper of type $(D_t, S_x, D_y)$ if and only if $r$ does. By Theorem 4.21 in~\cite{chen2019ISSAC} or Theorem 4.43 in~\cite{chen2020JSC},
we have $r$ has a telescoper of type $(D_t, S_x, D_y)$ if and only if for each $i$ with $1\leq i \leq I$, either $\alpha_i$ is $D_t$-separable in $\overline{\bF(t, x)}$
or $\beta_i \in \overline{\bF(t)}$ and $\alpha_i\in \bF(t, x)(\beta_i)$ has a telescoper of type $(D_t, S_x)$.
The existence problem of telescopers of type $(D_t, S_x)$ in $\bF(t, x)(\beta)$ with $\beta \in \overline{\bF(t)}$ has been solved in~\cite{ChenSinger2012}.
To completely solve Problem~\ref{PROB:DSD}, it remains to solve the following separability problem.

\begin{problem}\label{PROB:alg}
Given an algebraic function $f(t, \vx)$ over $\bF(t, \vx)$, decide whether $f(t, \vx)$ is $D_t$-separable.
%, i.e., whether there exists a nonzero operator $L\in \bF(t)\langle D_t \rangle$ such that $L(f) = 0$.
\end{problem}
We assume that $\bF$ is an algebraically closed and computable subfield of $\bC$ in the remaining part of this section.

\subsection{A descent theorem}  \label{SUBSECT:descent}
We first recall some basic notions and results from the theory of algebraic functions of one variable~\cite{chevalley}.
Let $k$ be a field of characteristic zero and $k(x,y)$ be an algebraic function field of one variable over $k$, i.e., the transcendence degree of $k(x, y)$ over $k$ is one. This
means there exists a polynomial $f\in k[X, Y]$ such that $f(x, y)=0$. The field of constants of $k(x,y)$ is defined as the set of elements of $k(x,y)$
which are algebraic over $k$. A subring $R$ of $k(x, y)$ is called a \emph{valuation ring} if $k\subset R \subsetneqq k(x,y)$ and for any $x\in k(x, y)$, either $x\in R$ or $x^{-1}\in R$. Any valuation ring $R$ of
$k(x, y)$ is a local ring, whose unique maximal ideal $\frakp$ is called a \emph{place} of $k(x, y)$ and the quotient field $R/\frakp$ is called the \emph{residue field} of the place $\frakp$, denoted by $\Sigma_{\frakp}$.
\begin{lem}
\label{LM-uniqueplace}
Let $k(x,y)$ and $f\in k[X, Y]$ be as above.  Assume that $(\bar{x}, \bar{y})\in k^2$ satisfies that $f(\bar{x},\bar{y})=0$ and $\frac{\partial f}{\partial Y}(\bar{x},\bar{y})\neq 0$.
Then there is a unique place $\frakp$ of $k(x,y)$ containing $x-\bar{x}$ and $y-\bar{y}$. Furthermore, the residue field $\Sigma_{\frakp}$ of $\frakp$ is isomorphic to $k$ and $k$ is the field of constants of $k(x,y)$.
\end{lem}
\begin{proof}
By Corollary 2 of~\cite[page 8]{chevalley}, there is a place of $k(x,y)$ containing $x-\bar{x}$ and $y-\bar{y}$, say $\frakp$. Let $\fraka$ be the discrete valuation ring (DVR) with respect to $\frakp$. It is easy to see that the ring $k[x,y]$ is contained in $\fraka$. Let $\frakm$ be the ideal in $k[x,y]$ generated by $x-\bar{x}$ and $y-\bar{y}$. Then $\frakm$ is a maximal ideal. Denote by $R$ the localization of $k[x,y]$ at $\frakm$ and we still use $\frakm$ to denote the unique maximal ideal of $R$. Rewriting $f(X,Y)$ as a polynomial in $X-\bar{x}, Y-\bar{y}$ yields that
\[
  \left(\frac{\partial f}{\partial Y}(\bar{x},\bar{y})+(Y-\bar{y})A \right)(Y-\bar{y})+(X-\bar{x})B
\]
for some $A,B\in k[X-\bar{x},Y-\bar{y}]$.
Since $\frac{\partial f}{\partial Y}(\bar{x},\bar{y})\neq 0$, one has that $\frac{\partial f}{\partial Y} (\bar{x},\bar{y})+(y-\bar{y})A(x-\bar{x},y-\bar{y})$ is invertible in $R$ and so $y-\bar{y}\in (x-\bar{x})R$. It implies that $R$ is a regular local ring, i.e., a DVR. Therefore $R=\fraka$, since $R\subset \fraka$. This concludes that $\frakp$ is unique.

We have that $\Sigma_{\frakp}=R/\frakm=k[x,y]/\frakm\cong k$. Since the field of constants of $k(x,y)$ is a subfield of $\Sigma_{\frakp}$ under the natural homomorphism, it coincides with $k$.
\end{proof}
\begin{remark}
\label{RM-series}
   Let $k(x,y)$ and $(\bar{x},\bar{y})$ be as in Lemma~\ref{LM-uniqueplace}. The above proof implies that $k(x,y)$ can be embedded into the field of formal Laurent series $k((x-\bar{x}))$.
\end{remark}

\begin{thm}\label{THM:descent} Let $\bF\subseteq k \subseteq \CX$ be fields with $\bF$ being algebraically closed. Let $f(t,Y)$ be an irreducible polynomial  in $k[t,Y]$.  Let $k(t,y)$ be the  quotient field of $k[t,Y]/\langle f\rangle$. Assume that
\begin{enumerate}
\item\label{condition1} the places of $k(t)$ that ramify in $k(t,y)$ are defined over $\bF$, {i.e., their uniformizing parameters can be chosen to be $1/t$ or $t-c$ with $c\in \bF$}.
\item\label{condition2} there exists a solution $(a, \alpha)$ of the system
\begin{eqnarray*}
f(a,\alpha) &=& 0  \label{eqn1}, \\
\frac{\d f}{\d Y}(a,\alpha) & \neq & 0 \label{eqn2},
\end{eqnarray*}where $a \in \bF$ and $\alpha \in k$.
\end{enumerate}
Then there exists $\beta\in \overline{\bF(t)}$ such that $k(t,y)= k(t,\beta)$.
\end{thm}
\begin{proof}
{ Since $(a,\alpha)$ is a simple point of $f(t,Y)=0$ in $k^2$}, by \cite{ragot}, $f(t,Y)$ is absolutely irreducible over $k$. {This implies that $f$ is irreducible over $\CX$, i.e., $\CX[t,Y]/\langle f \rangle$ is an integral domain}. Let $\CX(t,y)$ be the quotient field of $\CX[t,Y]/\langle f \rangle$. Then $k(t,y)$ can be considered as a subfield of $\CX(t,y)$ under the natural homomorphism. From Theorem 3 in~\cite[page 92]{chevalley}, none of places of $\CX(t,y)$ is ramified with respect to $k(t,y)$. Therefore the condition 1 holds for $\CX(t,y)$. Proposition 2.1 in~\cite[page 10]{malle_matzat} states that there is $\beta \in \overline{\bF(t)}$ such that  $\CX(t,y)=\CX(t,\beta)$. Now there are $g_0(t),\cdots,g_{n-1}(t) \in\CX(t)$ such that
\begin{equation}
\label{EQ-lineardependence}
       \beta=\sum_{i=0}^{n-1} g_i(t)y^i,
\end{equation}
where $n=[\CX(t,y):\CX(t)]$. For each $i$, let $g_i = {q_i}/{q}$ with $q_i, q \in \CX[t]$ and let $s = \max_i\{\deg_tq_i, \deg_tq\}$. Equation~(\ref{eqn1}) implies that
$q\beta = \sum_{i=0}^{n-1} q_i y^i$
and therefore the set
\[\left\{t^j\beta,\, \,  t^jy^i\right\}_{j = 0, \ldots s,i = 0, \ldots n-1}\]
is linearly dependent over $\CX$.  This set lies in $k(t,y,\beta)$ and, since it is linearly dependent over $D_t$-constants in a larger differential field, it is linearly dependent over $D_t$-constants in $k(t,y,\beta)$. Denote by $\tilde{k}$ the set of $D_t$-constants of $k(t,y,\beta)$. If $\tilde{k}=k$, then $\beta\in k(t,y)$, which will conclude the proposition. Therefore it suffices to prove that $\tilde{k}=k$. It is easy to verify that $\tilde{k}$ coincides with the field of constants of $k(t,y,\beta)$. In the following, we will show that the field of constants of $k(t,y,\beta)$ is equal to $k$.

From Remark~\ref{RM-series}, $k(t,y)$ and $\CX(t,y)$ can be embedded into $k((t-a))$ and $\CX((t-a))$ respectively. We will consider them as the subfields of $k((t-a))$ and $\CX((t-a))$ respectively. Since $\beta\in \CX(t,y)\cap \overline{\bF(t)}$, $\bF$ is algebraically closed and $a\in \bF$, $\beta \in \bF((t-a))$. Therefore, $k(t,y,\beta)\subseteq k((t-a))$. Since $k$ is algebraically closed in $k((t-a))$, the field of constants of $k(t,y,\beta)$ is equal to $k$. This completes the proof. \end{proof}

 \subsection{Separability criteria }  \label{SUBSECT:sepalg}

Let $P = \sum_{i=0}^n A_i Y^i \in \bF(t, \vx)[Y]$ be the minimal polynomial of $y\in \overline{\bF(t, \vx)}$. We can always pick $(a, \alpha) \in \bF \times \overline{\bF(\vx)}$ such that
\begin{equation}
 \label{EQ:simplepoint}
   A_n(\vx,a)\neq 0,\,\,P(\vx,a, \alpha)=0\,\,\mbox{and}\,\,\frac{\partial P}{\partial Y}(\vx,a,\alpha)\neq 0.
 \end{equation}
Let $K=\bF(\vx,\alpha)$ and $\ell=[K(t, y): K(t)]$.
Asume that $z\in \overline{\bF(t, \vx)}$ also satisfies the equation $P(z)=0$.  Then $z$ and $y$ are conjugated over $\bF(t, \vx)$. By Theorem 3.2.4 in~\cite{BronsteinBook}, any field automorphism of the splitting
field of $P$ commutes with the derivation $D_t$. So for any $L \in \bF(t)\langle D_t \rangle$, $L(z)=0$ if and only if $L(y)=0$. Thus to detect if there is a nonzero $L\in \bF(t)\langle D_t \rangle$ such that $L(y)=0$, it suffices to detect if there exists such operator for $z$. In the following, we will characterize all possible  $D_t$-separable algebraic functions.

Assume that $y$ is $D_t$-separable, i.e., there exists a nonzero $L \in \bF(t)\langle D_t\rangle$ such that $L(y)=0$. Let $\frakp$ be a place of $K(t)$ and $\frakq$ a place of $K(t,y)$ that is ramified with respect to $\frakp$. Suppose that $p$ and $q$ are uniformizing parameters of $\frakp$ and $\frakq$ respectively, and $e$ is the corresponding ramification index. Then $p=a q^e$ for some invertible $a$ in the DVR with respect to $\frakq$.  Furthermore assume that $p$ is an irreducible polynomial in $K[t]$. Let $\wp$ be a place of $\CX(t,y)$ lying above $\frakq$. Then by Theorem 3 in~\cite[page 92]{chevalley}), $\wp$ is not ramified with respect to $\frakq$ and so $q$ is a uniformizing parameter of $\wp$. Since $p\in \wp$, the uniformizing parameter of $\wp\cap \CX(t)$ can be selected as a factor of $p$, say $t-c$ for some $c\in \CX$. It is easy to see that $p/(t-c)$ is an invertible element in the DVR with respect to $\wp$. It implies that $t-c=\bar{a}q^e$ for some invertible element $\bar{a}$ {and thus $K(t,y)$ can be embedded into $\CX((t-c)^{1/e})$}. Therefore $y\in \CX((t-c)^{1/e}))$ and $c$ is a singular point of $L$. Note that the singular points of $L$ lie in the algebraically closed field $\bF$. So $c\in \bF$ and then $p=b(t-c)$ for some $b\in K$. {In other words, $t-c$ is a uniformizing parameter of $\frakp$}. Hence $K(t,y)$ satisfies the condition 1 of Theorem~\ref{THM:descent}. By Theorem~\ref{THM:descent}, there is $\beta \in \overline{\bF(t)}$ such that $K(t,y)=K(t,\beta)$.
We now characterize separable algebraic functions as follows.
\begin{prop}
\label{PROP:form}
Let $P = \sum_{i=0}^n A_i Y^i\in \bF[t, \vx][Y]$ with $A_n \neq 0$ be the minimal polynomial of $y \in \overline{\bF(t, \vx)}$. Let $K = \bF(\vx)(\alpha)$ with $\alpha\in \overline{\bF(\vx)}$ be as in~\eqref{EQ:simplepoint} and $\beta \in \overline{k(t)}$ be such that $K(t,y) = K(t, \beta)$. If $y$ is $D_t$-separable, then
  \begin{enumerate}
  \item [$(1)$]$A_n(\bfx,t)$ is split, i.e., $A_n(\bfx,t)=a(\bfx)b(t)$, where $a(\bfx)\in \bF[\bfx], b(t)\in \bF[t]$, and
  \item [$(2)$]
   \begin{equation}
   \label{EQ:form}
        y=\frac{1}{b(t)q(t)}\sum_{i=0}^{\ell-1} a_i(t)\beta^i,
    \end{equation}
   where $\ell=[K(t,y):K(t)], a_i(t)\in K[t]$ and $q(t)$ is the discriminant of the base $\{1,\beta,\cdots,\beta^{\ell-1}\}$.
 \end{enumerate}
\end{prop}
\begin{proof}
Let $r_i = A_i/A_n = p_i/q_i \in \bF(t, \vx)$ with $0\leq i\leq n$, $p_i, q_i \in \bF[t, \vx]$ and $\gcd(p_i, q_i) = 1$. Since $y$ is $D_t$-separable, so are all of the conjugate roots of $P(Y)=0$.
By Vieta's formulas, the $r_i$'s are polynomials of these roots, which therefore are also $D_t$-separable by Proposition~\ref{PROP:closure}.
By Theorem~\ref{THM:seprat}, $q_i$ is split for all $i$ with $0\leq i\leq n$. Since $A_n$ is the LCM of the $q_i$'s, we have $A_n(\bfx,t)$ is also split.

Let $S$ be the integral closure of $K[t]$ in $K(t,y)$. Then $\beta, A_n(\bfx,t)y\in S$. Since $\{1,\beta,\cdots,\beta^{\ell-1}\}$ is a base of $K(t,y)$ over $K(t)$, one has that
$$
   A_n(\bfx,t)y=\frac{1}{q(t)}\sum_{i=0}^{\ell-1}g_i(t)\beta^i,
$$
where $g_i(t)\in K[t]$. Setting $a_i(t)=g_i(t)/a(\bfx)$, we obtain the required expression for $y$.
\end{proof}

Recall that $K=\bF(\vx,\alpha)$ and $\ell = [K(t, y) : K(t)]$. Since the $i$-th derivative of $y$ is also in $K(t, y)$ for any $i\in \bN$, we have that $Y=(1,y,y^2,\cdots, y^{\ell-1})^t$ satisfies a
linear differential system of the form
\begin{equation}
\label{EQ:associateddifferential}
Y'=AY, \quad \text{where $A \in \text{Mat}_\ell(K(t))$}.
\end{equation}
We will call (\ref{EQ:associateddifferential}) the associated differential equation of $y$ over $K(t)$. The following proposition will allow us to design an algorithm for testing the separability of algebraic functions.
\begin{prop}
\label{PROP:criterion}
   Let $y$ and $K$ be as above. Assume that (\ref{EQ:associateddifferential}) is the associated differential equation of $y$ over $K(t)$. Then $y$ is $D_t$-separable if and only if there is an invertible matrix $G$ with entries in $K[t]$ such that
   $$
      G^{-1}G'- G^{-1}AG\in \Mat_\ell(\bF(t)).
   $$
\end{prop}

\begin{proof}
Assume that there exists a nonzero $L\in \bF(t)\langle D_t \rangle$ such that $L(y)=0$. Then by Proposition~\ref{PROP:form}, $y$ has the form (\ref{EQ:form}). Let $E$ be the Galois closure of $K(t,\beta)$ over $K(t)$. Let $\beta_1=\beta, \beta_2,\cdots, \beta_\ell$ be the conjugates of $\beta$ and $\sigma_i\in \Gal(E/K(t))$ such that $\sigma_i(\beta)=\beta_i$. Then $\sigma_1(y),\cdots,\sigma_{\ell}(y)$ are all zeroes of $P(\bfx,t,y)$.
We will denote the Vandermonde matrix generated by $\sigma_1(y),\cdots, \sigma_\ell(y)$ by $U(y)$  and the one generated by $\beta_1,\cdots, \beta_\ell$ by $U(\beta)$.
Then $U(y)$ is a fundamental matrix of the system (\ref{EQ:associateddifferential}) and $U(\beta)$ is a fundamental matrix of a system $Y'=BY$ with $B\in \Mat_\ell(\bF(t))$.
Using the argument similar to that in the proof of Proposition~\ref{PROP:criterion}, we have that
for all $j$ with $1\leq j\leq \ell-1$,
\begin{equation}
 \label{EQ:powers1}
  y^j=\frac{1}{b(t)^jq(t)}\sum_{i=0}^{\ell-1} a_{i,j}(t)\beta^i,
 \end{equation}
where $a_{i,j}(t)\in K[t]$ and $b(t),q(t)$ are as in Proposition~\ref{PROP:form}.
 Applying $\sigma_l$ to both sides of the equalities (\ref{EQ:powers1}) implies that
 \begin{equation}
 \label{EQ:powers2}
     \sigma_l(y)^j=\frac{1}{b(t)^jq(t)}\sum_{i=0}^{\ell-1} a_{i,j}(t)\beta^i_l ,
 \end{equation}
 where $j=1,\cdots, \ell-1, l=1,\cdots,\ell$.
 Let $\tilde{a}_{i,j}=a_{i,j}b^{\ell-1-j}$ and
 \[
       G=
       \begin{pmatrix} b(t)^{\ell-1}q(t) & 0 & \cdots & 0 \\  \tilde{a}_{0,1}(t) & \tilde{a}_{1,1}(t) &  \cdots & \tilde{a}_{\ell-1,1}(t) \\ \vdots & \vdots & \vdots & \vdots \\ \tilde{a}_{0,\ell-1}(t) & \tilde{a}_{1,\ell-1}(t) &  \cdots & \tilde{a}_{\ell-1,\ell-1}(t)\end{pmatrix}
 \]
 that is an element in $\Mat_\ell(K[t])$.
 Then the equations (\ref{EQ:powers2}) can be rewritten as $U(y)=({G}U(\beta))/({b(t)^{\ell-1}q(t)})$. Hence $G$ is invertible and an easy calculation yields that
\begin{align*}
     U(\beta)'&=(b^{\ell-1}qG^{-1}U(y))'\\
     &=\left((b^{\ell-1}q)'-b^{\ell-1}qG^{-1}G'+b^{\ell-1}qG^{-1}AG\right)G^{-1}U(y)\\
     &=BU(\beta)=b^{\ell-1}qBG^{-1}U(y).
\end{align*}
This implies that
$$
    G^{-1}AG-G^{-1}G'=B-\frac{(b^{\ell-1}q)'}{b^{\ell-1}q} \in \Mat_\ell(\bF(t)).
$$

Now we prove the converse. Assume that there is an invertible matrix $G\in \Mat_\ell(K[t])$ such that
$$
    \tilde{B}=G^{-1}AG-G^{-1}G' \in \Mat_\ell(\bF(t)).
$$
Then $U(y)=G F$, where $F$ is a fundamental matrix of $Y'=\tilde{B}Y$ with entries in some differential extension field of $K(t)$.
Obviously, the entries of both $G$ and $F$ are annihilated by nonzero operators in $\bF(t)\langle D_t \rangle$
and thus so are the sum of products of entries of $G$ and $F$, in particular, so is $y$.
\end{proof}
\begin{remark}
Once $\beta$ is computed, one can obtain the linear differential equations $Y'=BY$ satisfied by $U(\beta)$.
\end{remark}

\subsection{An algorithm for testing separability}  \label{SUBSECT:algo}
We now present an algorithm to decide whether a given algebraic function $y\in \overline{\bF(t, \vx)}$
is $D_t$-separable or not. For the sake of simplicity, we may take $\bF = \bar \bQ$, the field of all algebraic numbers over $\bQ$. Let $P=\sum_{i=0}^n A_i Y^i \in \bF[t,\vx][Y] $ be the minimal polynomial of $y$.
Furthermore, assume that $A_n$ is split. Under this assumption, $y$ is $D_t$-separable if and only if  $A_n y$ is $D_t$-separable. Therefore without loss of generality, we may assume that
\begin{equation}
\label{EQ:algebraicfunction2}
P(\bfx,t,Y)=Y^n+A_{n-1}(\bfx,t)Y^{n-1}+\cdots+A_0(\bfx,t),
\end{equation}
where $A_i \in \bF[\bfx,t]$. Let $(a,\alpha)\in \bF\times \overline{\bF(\vx)}$ satisfy
\begin{equation}
\label{con2}
  P(\bfx,a,\alpha)=0, \,\,\frac{\partial P}{\partial Y}(\bfx,a,\alpha)\neq 0,
\end{equation}
and let $K=\bF(\vx,\alpha)$.
Then $P(\vx,t,Y)$ may be factorized into a product of irreducible polynomials in $K[t,Y]$. There is a unique factor of $P(\vx,t,Y)$ in $K[t,Y]$
vanishing at $(a,\alpha)$, denoted by $\bar{P}(\vx,\alpha,t,Y)$. Let $K(t,y)$ be the quotient field of $K[t,Y]/\langle \bar{P}(\vx,\alpha,t,Y)\rangle$. Furthermore suppose that $K(t,y)$ satisfies the condition 1 of Theorem~\ref{THM:descent}. Then Theorem~\ref{THM:descent} implies that there is $\beta \in \overline{\bF(t)}$ such that $K(t,y)=K(t,\beta)$. We shall show how to find such $\beta$.

Let $R=\bF(t)[\vx]$ and $S$ the integral closure of $R$ in $K(t,y)$. Then $\alpha, y\in S$.
 Suppose that
\begin{equation}
\label{EQ:minimalpoly}
     \bar{P}(\vx,\alpha,t,Y)=B_\ell Y^\ell+B_{\ell-1}Y^{\ell-1}+\cdots+B_0,
\end{equation}
where $B_\ell\in \bF[\bfx], B_i\in k[\bfx,\alpha,t]$ with $i=0,\cdots, \ell-1$. Note that
\begin{align*}
    [K(t,y):\bF(\vx,t)]&=[K(t,y):K(t)][K(t):\bF(\vx,t)]\\
    &=[K(t,y):K(t)][K:\bF(\vx)]=\ell [K:\bF(\vx)].
\end{align*}
The set
\[
  \left\{ \alpha^i y^j\left| i=0,\cdots, [K:\bF(\vx)]-1, j=0,\cdots,\ell-1\right.\right\}
\]
is a base of $K(t,y)$ over $\bF(\vx,t)$. Let $D(\vx,t)$ be the discriminant of the above base and let $F(\vx,Y)$ be an irreducible polynomial in $ \bF[\vx,Y]$ such that $F(\vx,\alpha)=0$. Then we have
\begin{lem}
\label{LM-beta}
Let $(\bfc,b) \in \bF^{m+1}$ satisfy $F(\bfc, b)=0$ and $D(\bfc,t)B_\ell(\bfc)\neq 0$. Then $\bar{P}(\bfc,b, t,Y)$ is irreducible in $\bF[t,Y]$ and for any
root $Y=\gamma$ of $\bar{P}(\bfc,b, t,Y) = 0$,
we have that $K(t,y)$ is isomorphic to $K(t,\gamma)$.
\end{lem}

\begin{proof} Let $\beta \in K(t,y)$ be as above.  Since $\beta $ is algebraic over $\bF(t)$ we have that $\beta$ is integral over $R=\bF(t)[\bfx]$. Therefore we may write
 \[\beta = \frac{1}{D(\vx,t)}\sum b_{i,j}\alpha^i y^j,\]
 where the $b_{i,j} \in R$.  Let $(\bfc,b)$ satisfy the hypothesis of the lemma and consider the ideal
 \[\frakp = \langle x_1-c_1, \ldots ,x_m - c_m,\alpha-b\rangle \lhd R[\alpha].\]
 Note that $\frakp$ is a maximal ideal. The Going Up Theorem implies that there is a maximal ideal $\frakq \lhd S$ such that $\frakq\cap R[\alpha] = \frakp$.  In particular, $D(\vx,t) \notin \frakq$. There is a natural map $\phi : S \rightarrow S/\frakq$. We will let $M$ denote the field $S/\frakq$.  The element $\gamma = \phi(y)$ is a root of $\bar{P}(\bfc,b,t,\gamma) = 0$.  Since the minimal polynomial $Q(t,Y)$ of $\beta$ lies in $\bF[t,Y]$, it remains unchanged when we apply $\phi$ to its coefficients.  Therefore $\phi(\beta)$ satisfies $Q(t,\phi(\beta)) = 0$.  In particular, the degree of $\phi(\beta)$ over $\bF(t)$ is equal to $\ell$, the degree of $K(t,\beta)$ over $K(t)$. Since
 \[ \phi(\beta) = \frac{1}{D(\bfc,t)}\sum \phi(b_{i,j})\phi(\alpha)^i \gamma^j\]
 we have that $\phi(\beta) \in \bF(t)(\gamma)$. Note that $\bar{P}(\bfc,b,t,Y)\neq 0$. The element $\gamma$ satisfies $\bar{P}(\bfc,b,t,\gamma) = 0$
 and so it has degree at most $\ell$ over $\bF(t)$. Since $\phi(\beta) \in \bF(t,\gamma)$, we have that
\begin{align*}
  \ell \geq [\bF(t,\gamma):\bF(t)] \geq [\bF(t,\phi(\beta)):\bF(t)]&=[K(t,\beta):K(t)]\\
  &=[K(t,y):K(t)]=\ell
\end{align*}
and so $[\bF(t,\gamma):\bF(t)]  = \ell$ .  Therefore $\bar{P}(\bfc,b,t,Y)$ is irreducible. Furthermore $\bF(t,\beta) $ is isomorphic to $\bF(t, \phi(\beta)) = \bF(t,\gamma)$. This implies that $K(t,y) $  is isomorphic to $K(t,\gamma)$.
\end{proof}
Let $\bar{P}(\bfx,\alpha,t,Y)$ be as above. Lemma~\ref{LM-beta} implies that if $y$ is $D_t$-separable then one can compute $(\bfc,b)\in \bF^{m+1}$ such that $\bar{P}(\bfc,b,t,Y)$ is irreducible over $\bF(t)$ and $\beta$ can be taken to be a zero of $\bar{P}(\bfc,b,t,Y)$. From $\bar{P}(\bfc,b,t,Y)$, we can construct the associated differential equation of $\beta$ over $\bF(t)$. Denote this associated differential equation by $Y'=BY$ with $B\in \Mat_\ell(\bF(t))$. The proof of Proposition~\ref{PROP:criterion} implies that if $y$ is $D_t$-separable then there is an invertible matrix $G$ with entries in $K[t]$ such that
\[
   G'=AG-G\left(B-\frac{q'(t)}{q(t)}\right),
\]
where $q(t)$ is the discriminant of $\{1,\beta,\cdots,\beta^{\ell-1}\}$ and $Y'=AY$ is the associated differential equation of $y$ over $K(t)$. { Here the polynomial $b(t)$ in (\ref{EQ:powers1}) disappears because we assume that $P$ is monic in $Y$.}
Note that $G$ is a polynomial solution of the linear differential equation $Y'=AY-Y(B-q(t)'/q(t))$, which can be computed by algorithms developed in~\cite{abramov-bronstein, barkatou}.

We summarize the above results as the following algorithm.
\begin{algorithm}
Input: An irreducible polynomial
$$
   P(t,\vx,Y)=A_n Y^n+A_{n-1} Y^{n-1}+\cdots+A_0 \in \bF[t,\vx, Y].
$$
Output: ``Yes" if $y$ is $D_t$-separable, otherwise ``No", where $y\in \overline{k(\bfx,t)}$ is a root of $P(Y)=0$.
\begin{itemize}
   \item [$(1)$] If $A_n $ is not split, then $y$ is not $D_t$-separable and return ``No".
   \item [$(2)$] Transform $P(\bfx,t,Y)$ into a monic polynomial by replacing $Y$ by $Y/A_n$  and clear the denominators.
   \item [$(3)$] Compute $\beta$:
      \begin{itemize}
         \item [$(3.a)$] Find $(a,\alpha)\in \bF\times \overline{\bF(\vx)}$ satisfying the conditions (\ref{con2}).
           \item [$(3.b)$] Decompose $P $ into a product of irreducible polynomials over $\bF(\bfx,\alpha)$. Let $\bar{P}(\bfx,\alpha, t,Y)$ be the irreducible factor satisfying that
                 \[\bar{P}(\bfx,\alpha, a,\alpha)=0.\]
              \item [$(3.c)$] Compute $D(\bfx,t)$, the discriminant of the base $\{\alpha^i \bar{y}^j\}$, where $\bar{y}$ is a zero of $\bar{P}(\bfx,\alpha, t,Y)$ in $\overline{\bF(t, \vx)}$.
              \item [$(3.d)$] Compute a point $(\bfc,b)\in \bF^{m+1}$ such that
              \[\text{$D(\bfc,t)B_\ell(\bfc)\neq 0$ and $F(\bfc,b)=0$,} \]where $F$ is the minimal polynomial of $\alpha$ over $\bF(\bfx)$ and $B_\ell(\bfx)$ is the leading coefficient of $\bar{P}(\bfx,\alpha,t,Y)$.
              \item [$(3.e)$] Let $\beta$ be a zero of $\bar{P}(\bfc,b,t,Y)=0$ in $\overline{\bF(t)}$.
      \end{itemize}
   \item [$(4)$] Compute $G$:
      \begin{itemize}
          \item [$(4.a)$]
             Compute $q(t)$, the discriminant of the base $\{\beta^j|j=0,\cdots,\ell-1\}$ and compute the associated differential equations of $y$ and $\beta$, which are denoted by $Y'=AY$ and $Y'=BY$ respectively.
          \item [$(4.b)$] By algorithms developed in \cite{abramov-bronstein, barkatou}, compute a base of polynomial solutions of { $Z'=AZ-Z(B-q(t)'/q(t))$, where $Z=(z_{ij})$ with indeterminate entries,} say $\{Q_1,\cdots,Q_s\}$.
          \item [$(4.c)$] Compute $C=\det(z_1Q_1+\cdots+z_sQ_s)$ with $z_1,\cdots,z_s$ being indeterminates. If $C=0$ then return ``No", otherwise return ``Yes".
       \end{itemize}
\end{itemize}
\end{algorithm}

We now show an example to illustrate the main steps of the above algorithm.

\begin{exam}
\label{exam:algebraiccase}
 Let $\bE=\bar{\QX}(t,x)$ and $y$ be the algebraic function over $\bE$ defined by
\[
  P(x,t,Y):=Y^2-2(xt+1)Y+(xt+1)^2-t.
\]
We are going to decide whether $y$ is $D_t$-separable or not. We will follow the above algorithm step by step. Since $P(x,t,Y)$ is monic in $Y$. We begin with the third step, i.e., computing $\beta$.
\begin{enumerate}
\item [(3)] Compute $\beta=\sqrt{t}+1$:
\begin{enumerate}
  \item [(3.a)] Set $(a,\alpha)=(1,x)$. One sees that
  $
    P(x,1,x)=0
  $
   and $$\frac{\partial P}{\partial Y}(x,1,x)=-2\neq 0.$$ So $\bar{\bQ}(x,\alpha)=\bar{\bQ}(x)$.
   \item [(3.b)] Since $P(x,t,Y)$ is irreducible over $\bar{\bQ}(x)$, we take $\bar{P}(x,\alpha,t,Y)$ to be $P(x,t,Y).$
   \item [(3.c)] Set $D(x,t)=4t$, which is the discriminant of the base $\{1,\bar{y}\}$ with $P(x,t,\bar{y})=0$.
   \item [(3.d)] One sees that $B_2(x)=1$ and $F=z-x$. So the point $(0,0)$ satisfies $D(0,t)B_2(0)\neq 0$ and $F(0,0)=0$.
   \item [(3.e)] Set $\beta=\sqrt{t}+1$ which is a zero of $P(0,t,Y)=Y^2-2Y+1-t$.
\end{enumerate}
\item [(4)] Compute $G$:
\begin{enumerate}
  \item [(4.a)] Set $q(t)=4t$, which is the discriminant of the base $\{1,\beta\}$, and set
  \[
     A=\begin{pmatrix}
       0 & 0 \\
       \frac{x}{2}-\frac{1}{2t} & \frac{1}{2t}
     \end{pmatrix},\,\,
     B=\begin{pmatrix}
          0 & 0 \\
          -\frac{1}{2t} & \frac{1}{2t}
     \end{pmatrix}.
  \]
   Then $Y'=AY$ and $Y'=BY$ are the associated differential equations of $y$ and $\beta$ respectively.
   \item [(4.b)] Set $Z=(z_{ij})_{1\leq i,j\leq 2}$, and compute a base of the polynomial solutions of the system
   {$Z'=AZ-Z(B-1/t)$.} One has that
   \[
      \left\{Q_1:=\begin{pmatrix}
               t & 0 \\
               xt^2+t & 0
             \end{pmatrix},\,\,Q_2:=\begin{pmatrix}
                                 0 & 0 \\
                                 -t & t
                               \end{pmatrix}\right\}
   \]
   is a required base.
   \item [(4.c)] One has that ${\rm det}(z_1Q_1+z_2Q_2)=z_1z_2t^2\neq 0$. So $y$ is $D_t$-separable.
\end{enumerate}
\end{enumerate}
\end{exam}

%\begin{notation}
%\label{notation}
% Let $K=k(\vx,\alpha)$ and $\ell=[K(t,y):K(t)]$.
%\end{notation}
%\section{D-finite and P-recursive cases}\label{SECT:hyper}

\section{Conclusion and future work}\label{SECT:conc}
We present a connection between the separability problems and the existence problems in creative telescoping.
Separability criteria are given for rational functions, hyperexponential functions, hypergeometric terms and
algebraic functions. Some results in the algebraic case have been generalized to the case of $D$-finite functions
whose annihilating operators of minimal order are completely reducible in~\cite{chen2021}. The existence problems of telescopers for
rational functions in three variables are now completely settled by combining the results in~\cite{chen2020JSC} with the
separability criteria in this paper.

In terms of future research, the first natural direction is to solve
the separability problem for P-recursive sequences, which may have
applications in solving the general termination problem of Zeilberger's algorithms beyond the
hypergeometric case. Another direction is to develop more symbolic computational tools for
the method of separation of variables for partial differential equations  as in~\cite{Miller1977}.

%\medskip

%\medskip
%\noindent {\bf Acknowledgement.} The authors would like to thank Hao Du, Ruyong Feng and Ziming Li for helpful discussions.
%\noindent {\bf Acknowledgement.} The authors thank Michael F.\ Singer for helpful discussions in the algebraic case.
%and the anonymous referees for their constructive and helpful comments.

\bibliographystyle{plain}

%{
%%\small
%%\scriptsize
%\bibliography{separable}
%}
%
%

%\bibliographystyle{plain}
%
%\bibliography{existence}

\begin{thebibliography}{10}

\bibitem{Abramov1971}
S.~A. Abramov.
\newblock The summation of rational functions.
\newblock {\em \v{Z}. Vy\v{c}isl. Mat i Mat. Fiz.}, 11:1071--1075, 1971.

\bibitem{Abramov1974}
S.A. Abramov.
\newblock Solution of linear finite-difference equations with constant
  coefficients in the field of rational functions.
\newblock {\em USSR Computational Mathematics and Mathematical Physics},
  14(4):247 -- 251, 1974.

\bibitem{Abramov1975}
Sergei~A. Abramov.
\newblock The rational component of the solution of a first order linear
  recurrence relation with rational right hand side.
\newblock {\em \v Z. Vy\v cisl. Mat. i Mat. Fiz.}, 15(4):1035--1039, 1090,
  1975.

\bibitem{Abramov1995b}
Sergei~A. Abramov.
\newblock Indefinite sums of rational functions.
\newblock In {\em ISSAC '95: Proceedings of the 1995 {I}nternational
  {S}ymposium on Symbolic and {A}lgebraic {C}omputation}, pages 303--308, New
  York, NY, USA, 1995. ACM.

\bibitem{abramov-bronstein}
Sergei~A. Abramov and Manuel Bronstein.
\newblock On solutions of linear functional systems.
\newblock In {\em ISSAC'01: Proceedings of the 2001 {I}nternational {S}ymposium
  on {S}ymbolic and {A}lgebraic {C}omputation}, pages 1--6, New York, NY, USA,
  2001. ACM.

\bibitem{AbramovLeLi2005}
Sergei~A. Abramov, Ha~Quang Le, and Ziming Li.
\newblock Univariate {O}re polynomial rings in computer algebra.
\newblock {\em J. of Mathematical Sci.}, 131(5):5885--5903, 2005.

\bibitem{Almkvist1990}
Gert Almkvist and Doron Zeilberger.
\newblock The method of differentiating under the integral sign.
\newblock {\em J. Symbolic Comput.}, 10:571--591, 1990.

\bibitem{barkatou}
Moulay~A. Barkatou.
\newblock On rational solutions of systems of linear differential equations.
\newblock {\em J. Symbolic Comput.}, 28(4-5):547--567, 1999.
\newblock Differential algebra and differential equations.

\bibitem{BronsteinBook}
Manuel Bronstein.
\newblock {\em {S}ymbolic {I}ntegration {I}: {T}ranscendental {F}unctions},
  volume~1 of {\em Algorithms and Computation in Mathematics}.
\newblock Springer-Verlag, Berlin, second edition, 2005.

\bibitem{BronsteinPetkovsek1996}
Manuel Bronstein and Marko Petkov{\v{s}}ek.
\newblock An introduction to pseudo-linear algebra.
\newblock {\em Theoret. Comput. Sci.}, 157:3--33, 1996.

\bibitem{chen2020JSC}
Shaoshi Chen, Lixin Du, Rong-Hua Wang, and Chaochao Zhu.
\newblock On the existence of telescopers for rational functions in three
  variables.
\newblock {\em J. Symbolic Comput.}, 104:494 -- 522, 2021.

\bibitem{chen2019ISSAC}
Shaoshi Chen, Lixin Du, and Chaochao Zhu.
\newblock Existence problem of telescopers for rational functions in three
  variables: The mixed cases.
\newblock In {\em ISSAC'19: Proceedings of the 2019 on {I}nternational
  {S}ymposium on {S}ymbolic and {A}lgebraic {C}omputation}, pages 82--89, New
  York, NY, USA, 2019. ACM.

\bibitem{Chen2014}
Shaoshi Chen, Ruyong Feng, Ziming Li, and Michael~F. Singer.
\newblock Parallel telescoping and parameterized {P}icard--{V}essiot theory.
\newblock In {\em ISSAC '14: Proceedings of the 2014 International Symposium on
  Symbolic and Algebraic Computation}, pages 99--106, New York, NY, USA, 2014.
  ACM.

\bibitem{chen2021}
Shaoshi Chen, Ruyong Feng, Ziming Li, Michael~F. Singer, and Stephen Watt.
\newblock Telescopers for differential forms with one parameter, 2021.
\newblock arXiv: 2101.06576.

\bibitem{ChenSinger2012}
Shaoshi Chen and Michael~F. Singer.
\newblock Residues and telescopers for bivariate rational functions.
\newblock {\em Adv. Appl. Math.}, 49(2):111--133, August 2012.

\bibitem{chevalley}
Claude Chevalley.
\newblock {\em Introduction to the {T}heory of {A}lgebraic {F}unctions of {O}ne
  {V}ariable}.
\newblock Mathematical Surveys, No. VI. American Mathematical Society, New
  York, NY, 1951.

\bibitem{Gosper1978}
Ralph~William Gosper, Jr.
\newblock Decision procedure for indefinite hypergeometric summation.
\newblock {\em Proc. Nat. Acad. Sci. U.S.A.}, 75(1):40--42, 1978.

\bibitem{Hermite1872}
Charles Hermite.
\newblock Sur l'int\'egration des fractions rationnelles.
\newblock {\em Ann. Sci. \'Ecole Norm. Sup. (2)}, 1:215--218, 1872.

\bibitem{LeLi04}
Ha~Quang Le and Ziming Li.
\newblock On a class of hyperexponential elements and the fast versions of
  {Z}eilberger's algorithm, 2004.
\newblock MM-Res. Preprints (2004) No. 23, 136-150.

\bibitem{malle_matzat}
Gunter Malle and B.~Heinrich Matzat.
\newblock {\em Inverse {G}alois theory}.
\newblock Springer Monographs in Mathematics. Springer-Verlag, Berlin, 1999.

\bibitem{Miller1977}
Willard Miller, Jr.
\newblock {\em Symmetry and {S}eparation of {V}ariables}.
\newblock Addison-Wesley Publishing Co., Reading, Mass.-London-Amsterdam, 1977.
\newblock With a foreword by Richard Askey, Encyclopedia of Mathematics and its
  Applications, Vol. 4.

\bibitem{Ostrogradsky1845}
Mikhail~Vasil'evich Ostrogradski{\u\i}.
\newblock De l'int{\'e}gration des fractions rationnelles.
\newblock {\em Bull.\ de la classe physico-math{\'e}matique de l'Acad.\
  Imp{\'e}riale des Sciences de Saint-P{\'e}tersbourg}, 4:145--167, 286--300,
  1845.

\bibitem{Paule1995}
Peter Paule and Markus Schorn.
\newblock A {M}athematica version of {Z}eilberger's algorithm for proving
  binomial coefficient identities.
\newblock {\em J. Symbolic Comput.}, 20(5-6):673--698, 1995.
\newblock Symbolic computation in combinatorics $\Delta{\sb{1}}$ (Ithaca, NY,
  1993).

\bibitem{PWZbook1996}
Marko Petkov{\v{s}}ek, Herbert~S. Wilf, and Doron Zeilberger.
\newblock {\em {$A=B$}}.
\newblock A. K. Peters Ltd., Wellesley, MA, 1996.
\newblock With a foreword by Donald E. Knuth.

\bibitem{ragot}
Jean-Fran\c{c}ois Ragot.
\newblock Probabilistic absolute irreducibility test for polynomials.
\newblock {\em J. Pure Appl. Algebra}, 172(1):87--107, 2002.

\bibitem{Stanley1980}
Richard~P. Stanley.
\newblock Differentiably finite power series.
\newblock {\em European J. Combin.}, 1(2):175--188, 1980.

\bibitem{Wilf1992}
Herbert~S. Wilf and Doron Zeilberger.
\newblock An algorithmic proof theory for hypergeometric (ordinary and
  ``{$q$}'') multisum/integral identities.
\newblock {\em Invent. Math.}, 108(3):575--633, 1992.

\bibitem{Zeilberger1991}
Doron Zeilberger.
\newblock The method of creative telescoping.
\newblock {\em J. Symbolic Comput.}, 11(3):195--204, 1991.

\end{thebibliography}

\end{document}